\documentclass[twocolumn]{autart}    

\usepackage[T1]{fontenc}
\usepackage{graphicx}          
\usepackage{amssymb,amsmath,bm}
\usepackage{mathrsfs}  
\usepackage[colorlinks=true,linkcolor=blue]{hyperref}
\newcommand{\R}{\mathbb R}

\DeclareMathOperator{\diag}{diag}

\newcommand{\norm}[1]{\left|\left| #1 \right|\right|}

\newcommand{\sign}{\textup{sign}}

\makeatletter

\newcommand{\DeclareIndependentEnv}[3]{%
  \@ifundefined{c@#1ctr}{\newcounter{#1ctr}}{}%
  \@ifundefined{#1}{%
    \newenvironment{#1}[1][]%
    {%
      \refstepcounter{#1ctr}%
      \par\medskip
      \noindent{\bfseries #2~\arabic{#1ctr}%
      \def\temp{##1}\ifx\temp\@empty\else\ (##1)\fi.}%
      #3\ %
    }%
    {\par\medskip}%
  }{%
    \renewenvironment{#1}[1][]%
    {%
      \refstepcounter{#1ctr}%
      \par\medskip
      \noindent{\bfseries #2~\arabic{#1ctr}%
      \def\temp{##1}\ifx\temp\@empty\else\ (##1)\fi.}%
      #3\ %
    }%
    {\par\medskip}%
  }%
}

\DeclareIndependentEnv{ass}{Assumption}{\itshape}
\DeclareIndependentEnv{thm}{Theorem}{\itshape}
\DeclareIndependentEnv{lem}{Lemma}{\itshape}
\DeclareIndependentEnv{cor}{Corollary}{\itshape}
\DeclareIndependentEnv{prop}{Proposition}{\itshape}
\DeclareIndependentEnv{rem}{Remark}{\itshape}

\makeatother

\renewenvironment{pf}
{\par\noindent{\bf Proof.}\ }
{\hfill\rule{1.5mm}{1.5mm}\par}

\begin{document}

\begin{frontmatter}

\title{Modular Sign Compensation for MIMO Systems with Unknown Control Direction: An Exact Nominal Recovery Approach\thanksref{footnoteinfo}}

\thanks[footnoteinfo]{This paper was not presented at any IFAC meeting.
Corresponding author: D. Guti\protect\'{e}rrez-Oribio.}

\author[Paestum]{Diego Guti\protect\'{e}rrez-Oribio}\ead{diego.gutierrez@ensta.fr},
\author[Paestum]{Ioannis Stefanou}\ead{ioannis.stefanou@ensta.fr}

\address[Paestum]{IMSIA (UMR 9219), CNRS, EDF, ENSTA Paris,
Institut Polytechnique de Paris, 828 Boulevard des Mar\protect\'{e}chaux,
91762 Palaiseau, France}

\begin{keyword}
Switching controllers; unknown control direction, adaptive control; analysis of systems with uncertainties; MIMO systems.
\end{keyword}

\begin{abstract} 
This paper addresses stabilization of MIMO systems with uncertain time-varying diagonal input direction. We propose a modular switching sign-compensation layer acting as an outer wrapper around a nominal controller. Unlike Nussbaum-type gains, monitoring functions, or binary adaptive mechanisms, the method uses only bounded sign changes that preserve the nominal control magnitude and its properties. The compensation layer uses adaptive variables built from nominal Lyapunov quantities to search for the unknown input-sign configuration based on schedulers. Two schedulers are developed: a vector scheduler, where each input channel explores its own sign compensation and admits an online trapping certificate, and a scalar pattern scheduler, where one variable visits all diagonal sign matrices and gives a design-time recovery guarantee on sufficiently long constant-sign intervals. Once the correct sign configuration is set, the actual closed loop coincides with the nominal closed loop and the original nominal stability property is recovered. The approach is illustrated on a flight roll-reversal problem, a visual-servoing benchmark, and an underground-reservoir control example motivated by human induced-seismicity mitigation. 
\end{abstract}

\end{frontmatter}

\section{Introduction}
\label{sec:Introduction}

The design of feedback controllers usually assumes that the control direction is
known. In SISO systems this information is the sign of the high-frequency gain (\textit{i.e.}, the sign of the leading coefficient of the numerator of its transfer function),
whereas in MIMO systems it is encoded in the structure and signs of the
high-frequency gain matrix. If this information is wrong, a stabilizing feedback
may act as destabilizing positive feedback. This issue is relevant in several situations, including actuator polarity
uncertainty, loss or reversal of control effectiveness
\cite{b:tao2004adaptive,b:gou2024adaptive},
uncertain control allocation and fault-tolerant control
\cite{b:johansen2013control,b:amin2019review},
and MIMO systems in which the physical effect of an actuator on the regulated
output is not known a priori \cite{b:wang2020mimo}.

The unknown-control-direction problem has a long history in adaptive control.
A. S. Morse formulated the question of whether a universal adaptive stabilizer could
be constructed without knowing the sign of the high-frequency gain
\cite{b:morse1983recent}. R. D. Nussbaum answered this question by introducing an
oscillatory adaptive gain whose integral alternates between arbitrarily large
positive and negative values \cite{b:nussbaum1983remarks}. This idea led to the
so-called Nussbaum-gain technique, which has been widely used
in adaptive stabilization, nonlinear backstepping, output regulation, neural and
fuzzy adaptive control, and fault-tolerant control; see, for example,
\cite{b:ye1998adaptive,b:zhang2000adaptive,b:liu2006global,b:chen2019nussbaum}. The subsequent result of
Willems and Byrnes established global adaptive stabilization in the absence of
high-frequency-gain sign information for a class of linear systems
\cite{b:willems1984global}. 

Despite their theoretical importance, Nussbaum-type methods modify the effective
control amplitude through an adaptive gain. Their analysis relies on the
unbounded oscillatory growth of this gain, and the resulting transients may be
large. Moreover, the extension to time-varying or multivariable control
directions is nontrivial. In \cite{b:chen2019nussbaum}, the authors showed that, for time-varying unknown control coefficients, a general Nussbaum function is not automatically sufficient and
additional structural properties are required. In the
MIMO case, the difficulty is even more pronounced because several adaptive gains
may interact destructively. Recently, \cite{b:wang2020mimo} addressed this issue by
constructing special multivariable Nussbaum functions and coordinated
periodical intervals, using the nonzero leading principal minors of the
high-frequency gain matrix while allowing their signs to be unknown. This result significantly relaxes classical MIMO adaptive
control assumptions, but it remains an adaptive backstepping design based on
Nussbaum-type gain modulation.

Other works replace Nussbaum functions with logic-based switching rules
driven by a Lyapunov performance index. In~\cite{b:WU2016298}, a
switching-type adaptive controller is proposed for nonlinear systems with
multiple unknown control directions. This result was extended
in~\cite{b:8249868} to a more general adaptive switching framework that
simultaneously covers global exponential and finite-time stability for
lower-triangular nonlinear systems. These results are powerful, but their
construction is intrinsically linked to a specific nonlinear structure of the system.
There, the Lyapunov function, virtual controls, adaptive laws, and supervisory switching signal are all obtained recursively from the triangular/backstepping structure. Thus, the switching mechanism is not a modular wrapper around a given nominal controller, but a co-designed component tied to a specific plant class and Lyapunov construction.

Another line of work avoids unbounded Nussbaum gains by using
switching, monitoring functions, or binary adaptive mechanisms. For example,
sliding-mode designs for uncertain nonlinear systems with unknown control
direction were proposed in~\cite{b:oliveira2007control} and later extended to
multivariable systems in~\cite{b:oliveira2010sliding}. Periodic switching
functions were also used for output-feedback tracking of SISO plants with
unknown, and even time-varying, control direction
\cite{b:oliveira2011output,b:oliveira2015global}. Related binary adaptive
and monitoring-function approaches were developed in
\cite{b:bartolini2009second,b:oliveira2016binary,b:teixeira2021binary}. These
methods avoid unbounded gain growth, but the switching mechanism is embedded in
the specific sliding-mode, reference-model, or adaptive-control architecture.
Thus, they do not directly provide a modular sign-compensation wrapper that can
be placed around an arbitrary nominal stabilizing controller while preserving
its magnitude and recovering its nominal Lyapunov properties.

The contribution of the present paper is a nominal-controller-preserving sign adaptation
mechanism for MIMO systems with diagonal uncertain control direction. Compared with
Nussbaum-gain methods, the proposed strategy avoids unbounded gain growth and preserves the nominal control magnitude. Compared with Lyapunov performance index, binary adaptive and
monitoring-function approaches, the proposed method is not tied to a particular
reference-model architecture, observer structure, relative-degree reduction, or
sliding-mode parameterization. Instead, it acts as an outer sign wrapper around
a large class of nominal controllers (with or without dynamic/integral extension) for which the Lyapunov inequalities of the nominal control are
used. Consequently, the nominal Lyapunov decay, and therefore the stability property certified by the nominal design, is recovered after the sign-compensation transient.

To the best of the authors' knowledge, the proposed approach differs from existing Nussbaum, monitoring-function, binary adaptive, and sliding-mode designs because it provides a modular diagonal-sign wrapper that: (i) leaves the nominal control magnitude unchanged, (ii) uses only a nominal Lyapunov dissipation estimate, (iii) preserves a degraded Lyapunov decay under arbitrary locally finite sign
variations when sign-mismatch terms are dominated by the nominal Lyapunov
dissipation, and (iv) recovers the nominal closed-loop behavior once the correct sign interval is trapped. The result is illustrated on a flight vehicle, a robotic visual-servoing example, and an induced-seismicity mitigation problem in underground reservoirs. 


\section{Problem statement}
\label{sec:Problem}

Consider the uncertain MIMO system
\begin{equation}
    \dot{\sigma}
    =
    B(\sigma,t)S(t)u+\phi(\sigma,t),
    \label{eq:mimo_system}
\end{equation}
where $\sigma\in\R^n$ is the state, $u\in\R^m$ is the control input, and $\phi \in \R^n$ is an external perturbation. The system is assumed to have at least as many inputs as states, \textit{i.e.}, $m\ge n$. The control coefficient is decomposed into two terms: $B(\sigma,t)\in\R^{n\times m}$ and $S(t)\in\R^{m\times m}$. The first one is defined as
\begin{equation}
  B(\sigma,t)=(I+\Delta(\sigma,t))B_0(\sigma,t)
  \label{eq:B}
\end{equation}
where $B_0\in\R^{n\times m}$ is a known matrix and has full row rank, while $\Delta(\sigma,t)\in\R^{n\times n}$ is an uncertain multiplicative
perturbation satisfying
\begin{equation}
    \norm{\Delta(\sigma,t)}\le \bar\delta ,
    \label{eq:Delta_bound_mimo}
\end{equation}
for some known constant $\bar\delta\ge0$. Since $B_0$ has full row rank, there
exists a right inverse $B_0^+\in\R^{m\times n}$ such that $B_0B_0^+=I_n$.

The matrix $S(t)=\diag(s_1(t),\ldots,s_m(t))$, $s_i(t)\in\{-1,+1\}$, represents a time-varying input-sign matrix assumed to be piecewise constant and unknown. Thus, the problem addressed here is diagonal actuator-sign uncertainty around a known full-row-rank input map, which is a subclass of a general unknown MIMO high-frequency gain matrix. When $S(t)=I_m$, there exists a nominal (dynamic) controller of the form
\begin{equation}
\begin{aligned}
    \bar u(\sigma,\nu)&=B_0^+f_1(\sigma,\nu),\\
    \dot\nu&=f_2(\sigma,\nu),
\end{aligned}
\label{eq:mimo_nominal_control}
\end{equation}
where $\nu\in\R^p$ is an auxiliary integral variable, and the terms $f_1:\R^n\times\R^p\to\R^n$, and
$f_2:\R^n\times\R^p\to\R^p$ are control functions; explicit time dependence, needed for tracking or feedforward compensation, is omitted to keep the notation light. The nominal controller is assumed to contain the feedback, feedforward, or set-valued robust term needed to make $(\sigma,\nu)=(0,0)$ an equilibrium of the nominal closed-loop system. In particular, nonvanishing matched perturbations, $\phi(\sigma,t)$, are allowed only when they are compensated in this nominal sense.

Therefore, the nominal closed-loop dynamics reduce to
\begin{equation}
\begin{aligned}
    \dot\sigma
    &=
    (I+\Delta(\sigma,t))f_1(\sigma,\nu)+\phi(\sigma,t),\\
    \dot\nu
    &=
    f_2(\sigma,\nu).
\end{aligned}
\label{eq:mimo_nominal_system}
\end{equation}
The following assumption formalizes the robustness property required from the
nominal controller.

\begin{assum}
\label{ass:mimo_nominal_lyap}
There exists a known positive definite and radially unbounded function
$V(\sigma,\nu) \in C^1(\R^n\times\R^p;\R_{\ge0})$, and a continuous positive definite function
$W(\sigma,\nu):\R^n\times\R^p\to\R_{\ge0}$ such that
\begin{equation}
    \nabla_\sigma V^\top
    \bigl((I+\Delta)f_1+\phi\bigr)
    +
    \nabla_\nu V^\top f_2
    \le
    -W(\sigma,\nu)
    \label{eq:mimo_nominal_lyap}
\end{equation}
for all $(\sigma,\nu)$, and all $t\ge0$.
\end{assum}

Thus, in the absence of sign uncertainty, the nominal feedback renders the
origin robustly stable with respect to the perturbations
$\Delta$ and $\phi$. The objective is to preserve boundedness under locally finite variations of
\(S(t)\), and to recover the nominal behavior on intervals where the unknown
sign matrix is constant and the correct compensation is set.

\section{Switching sign compensation control design}
\label{sec:Control}

We will propose a control design that keeps the nominal law unchanged and
adds only an outer sign compensation. More precisely, let
\begin{equation}
\begin{split}
    p(\sigma,\nu)&:=\nabla_\sigma V(\sigma,\nu),
    \\
    q(\sigma,\nu)&:=B_0^+f_1(\sigma,\nu)=\bar u(\sigma,\nu),
\end{split}
    \label{eq:mimo_p_q_def}
\end{equation}
and apply
\begin{equation}
    u(\sigma,\nu,\theta)=N(\theta)q(\sigma,\nu).
    \label{eq:mimo_adaptive_control}
\end{equation}
Here, $N(\theta)$ is a sign-compensation gain to be designed below. Two types of scheduler design for $N(\theta)$ will be presented: a vector scheduler and a scalar scheduler. In both cases, the wrapper generates a diagonal sign matrix $N(\theta)\in\mathscr S_m:=\{\diag(\pi_1,\ldots,\pi_m):\pi_i\in \{-1,+1 \} \}$. The difference lies in how this matrix is scheduled: the vector scheduler updates the signs componentwise, whereas the scalar scheduler uses one auxiliary variable to select one matrix from $\mathscr S_m$. Respectively, the adaptive variable $\theta$ is either a
vector, with one component per input channel, or a scalar that schedules full
sign patterns.

Denote by $b_i\in\R^n$ the $i$-th column of $B_0$. For each input channel,
define
\begin{equation}
    \psi_i(\sigma,\nu)
    :=
    |q_i|
    \left(
        |p^\top b_i|
        +
        \bar\delta \norm{p}\norm{b_i}
    \right),
    \quad i=1,\ldots,m,
    \label{eq:mimo_psi_i}
\end{equation}
and let
\begin{equation}
    \Psi(\sigma,\nu):=\sum_{i=1}^m\psi_i(\sigma,\nu).
    \label{eq:mimo_Psi}
\end{equation}
The function $\psi_i$ bounds the Lyapunov contribution of the $i$-th channel
when its sign is incorrectly compensated, while $\Psi$ is the corresponding
aggregate bound.

Depending on the scheduler design, two types of adaptation for the auxiliary variable $\theta$ are proposed. One is using the vector information of the nominal control separately, \textit{i.e.}, \eqref{eq:mimo_psi_i}:
\begin{equation}
    \dot\theta_i
    =
    \gamma_i\psi_i(\sigma,\nu),
    \quad
    \gamma_i>0,
    \quad
    i=1,\ldots,m .
    \label{eq:mimo_theta_i}
\end{equation}
The second type will use the merged information \eqref{eq:mimo_Psi} as
\begin{equation}
    \dot\theta
    =
    \gamma\Psi(\sigma,\nu),
    \quad
    \gamma>0,
    \quad
    \theta(0)\ge0 .
    \label{eq:scalar_theta}
\end{equation}

With \eqref{eq:mimo_adaptive_control}, the extended closed-loop system is
\begin{equation}
\begin{aligned}
    \dot\sigma
    &=
    (I+\Delta)B_0S(t)N(\theta)q(\sigma,\nu)+\phi(\sigma,t),\\
    \dot\nu
    &=
    f_2(\sigma,\nu),\\
    \dot\theta
    &=
    F_\theta(\sigma,\nu),
\end{aligned}
\label{eq:mimo_extended_closed_loop}
\end{equation}
where $F_\theta$ is either \eqref{eq:mimo_theta_i} or \eqref{eq:scalar_theta}. Since $V$ is $C^1$ and positive definite, $p(0,0)=0$ and therefore $F_\theta(0,0)=0$. When the signs are correctly compensated, the extended system has the equilibrium set $\mathcal E_\theta:=\{(0,0,\theta):\theta\in\R_{\ge0}^{d_\theta}\}$, where $d_\theta=m$ in the vector case and $d_\theta=1$ for the scalar one. 

Two additional assumptions are required now.
\begin{assum}\label{ass:smooth}
The closed-loop system \eqref{eq:mimo_extended_closed_loop} admits well-defined maximal solutions. If \(f_1\) and \(f_2\) are continuous, solutions are understood in the piecewise-Carath\'eodory sense, with arcs concatenated at the switching instants of \(S(\cdot)\) and \(N(\cdot)\). If additional discontinuities are present, for instance due to a sliding-mode implementation, solutions are understood in the Filippov sense \cite{b:filippov}. In either case, all differential inequalities used below hold for almost every
\(t\) along any trajectory. We also assume that whenever a bounded solution satisfies \(\int_{0}^{\infty}W(\sigma(t),\nu(t))dt<\infty\), then $(\sigma(t),\nu(t)) \to 0$, and that any maximal solution that remains in a compact subset of \(\R^{n}\times\R^p\times\R^{d_\theta}\) is extendable forward in time.
\end{assum}

\begin{assum}
\label{ass:mimo_Psi_bounds}
There exists a constant $c_M>0$ such that
\begin{equation}
    \Psi(\sigma,\nu)
    \le
    c_M W(\sigma,\nu)
    \label{eq:mimo_Psi_bounds}
\end{equation}
for all $(\sigma,\nu)\in\R^n\times\R^p$.
\end{assum}

\begin{rem}
\label{rem:Psi_bound_check}
Assumption~\ref{ass:mimo_Psi_bounds} can be checked from the nominal
controller. Indeed, it is sufficient to verify that there exist constants
$a_M,b_M>0$ such that $\sum_{i=1}^m |q_i|\,|p^\top b_i|\le a_MW$, $\sum_{i=1}^m |q_i|\norm{p}\norm{b_i}\le b_MW $. Then, by \eqref{eq:mimo_psi_i}, Assumption~\ref{ass:mimo_Psi_bounds} holds
with $c_M=a_M+\bar\delta b_M$.

For linear control this condition is immediate. For instance, let
$q(\sigma)=-K\sigma$, $V(\sigma)=\frac12\sigma^\top P\sigma$, with $P=P^\top>0$,
and suppose that the nominal dissipation satisfies $W(\sigma)=\sigma^\top
Q\sigma$, $Q=Q^\top>0$. Denoting by $k_i^\top$ the $i$-th row of $K$, one has
$    \Psi(\sigma)
    \le
    \sum_{i=1}^m
    \norm{k_i}
    \left(
        \norm{Pb_i}
        +
        \bar\delta\norm{P}\norm{b_i}
    \right)
    \norm{\sigma}^2
    \le
    c_M W(\sigma)$, with any $c_M$ larger than the preceding coefficient divided by
$\lambda_{\min}(Q)$.

For homogeneous sliding-mode controllers, the same verification can be
reduced to a compactness argument. The nominal Lyapunov function and the control components are homogeneous with compatible degrees (see \textit{e.g.}, \cite{b:6144710,b:CRUZZAVALA2017232,b:Mathey-Moreno-2024}). Hence, once the origin is excluded, the ratio $\Psi/W$ is
bounded on a unit level set of the homogeneous Lyapunov function. This gives
a finite constant $c_M$ and therefore verifies Assumption~\ref{ass:mimo_Psi_bounds}.
\end{rem}

Taking the time derivative of the nominal Lyapunov function along
\eqref{eq:mimo_extended_closed_loop} results in
$    \dot V
    =
    p^\top\bigl((I+\Delta)B_0S(t)N(\theta)q+\phi\bigr)
    +
    \nabla_\nu V^\top f_2 $. Adding and subtracting the nominal part $\nabla_\sigma V^\top
    \bigl((I+\Delta)f_1+\phi\bigr)$ gives
$    \dot V
    \le
    -W
    +
    p^\top(I+\Delta)B_0(SN-I_m)q $. Since $S$ and $N$ are diagonal sign matrices, $SN-I_m=\diag(d_1,\ldots,d_m)$,
where $d_i=s_iN_i-1\in\{0,-2\}$. Moreover, using \eqref{eq:Delta_bound_mimo} one can get
$    |q_i p^\top(I+\Delta)b_i|
    \le
    |q_i|
    \left(
        |p^\top b_i|
        +
        \bar\delta\norm{p}\norm{b_i}
    \right)
    =
    \psi_i$. Thus, in the worst case scenario with $\mathcal J(t):=\{i:s_i(t)N_i(\theta(t))=-1\}$, one obtains
\begin{equation}
    \dot V
    \le
    -W
    +
    2\sum_{i\in\mathcal J(t)}\psi_i
    \le
    -W+2\Psi .
    \label{eq:mimo_wrong_bound}
\end{equation}
In particular, whenever the signs are correctly compensated, namely
$S(t)N(\theta(t))=I_m$, the nominal dissipation is recovered:
\begin{equation}
    \dot V\le -W .
    \label{eq:mimo_correct_bound}
\end{equation}

These two latter bounds of the Lyapunov function derivative will be used for the scheduler design. The estimate \eqref{eq:mimo_wrong_bound} also gives a common degraded
stability property, which is independent of the way in which the
sign-compensation matrix is scheduled. 

\begin{prop}
\label{prop:common_degraded_stability}
Consider the extended closed-loop system \eqref{eq:mimo_extended_closed_loop} generated by either the vector adaptation \eqref{eq:mimo_theta_i} or the scalar adaptation \eqref{eq:scalar_theta}, under Assumptions~\ref{ass:mimo_nominal_lyap}--\ref{ass:mimo_Psi_bounds}. If
\begin{equation}
    2c_M<1 ,
    \label{eq:common_degraded_condition}
\end{equation}
holds, then, for any locally finite switching signal $S(\cdot)$ and for the switching
matrix $N(\cdot)$ generated by the corresponding scheduler, every solution
satisfies
\begin{equation}
    \dot V
    \le
    -\eta W(\sigma,\nu),
    \qquad
    \eta:=1-2c_M>0,
    \label{eq:common_degraded_dissipation}
\end{equation}
for almost all $t\ge0$. Consequently, $V$ is nonincreasing. Under the convergence condition included in Assumption~\ref{ass:smooth}, $(\sigma,\nu)$ tends to zero. If, in addition, the nominal estimate \eqref{eq:mimo_nominal_lyap} certifies exponential or finite-time stability through a comparison inequality of $W$ with $V$, then the same qualitative property holds for \eqref{eq:common_degraded_dissipation}, with the rate multiplied by $\eta$. Moreover, the auxiliary scheduler variables are bounded, nondecreasing, and converge to finite limits.
\end{prop}

\begin{pf}
The proof follows directly from the scheduler-independent estimate
\eqref{eq:mimo_wrong_bound}. Indeed, by Assumption~\ref{ass:mimo_Psi_bounds},
$    \dot V
    \le
    -W+2\Psi
    \le
    -(1-2c_M)W
    =
    -\eta W $, which proves \eqref{eq:common_degraded_dissipation}. Since $V$ is nonincreasing, $(\sigma,\nu)$ is bounded and \eqref{eq:common_degraded_dissipation} gives $\int_0^\infty W(\sigma(t),\nu(t))dt<\infty$; convergence follows from Assumption~\ref{ass:smooth}. The exponential and finite-time cases follow immediately when the nominal proof uses, respectively, $W\ge aV$ or $W\ge aV^\alpha$, $\alpha\in(0,1)$. It remains only to bound the auxiliary variables. In the vector case,
$\dot\theta_i=\gamma_i\psi_i\le\gamma_i\Psi\le\gamma_ic_MW$. Combining this
estimate with \eqref{eq:common_degraded_dissipation} gives
$    \theta_i(t)-\theta_i(0)
    \le
    \frac{\gamma_ic_M}{\eta}\bigl(V(0)-V(t)\bigr)
    \le
    \frac{\gamma_ic_M}{\eta}V(0)$, for $i=1,\ldots,m$, where the shorthand $V(t):=V(\sigma(t),\nu(t))$ will be used hereinafter. Thus, every $\theta_i$ is bounded. Since $\dot\theta_i\ge0$, each $\theta_i$
converges to a finite limit.

In the scalar case, $\dot\theta=\gamma\Psi\le\gamma c_MW$, and the same
argument gives
$    \theta(t)-\theta(0)
    \le
    \frac{\gamma c_M}{\eta}\bigl(V(0)-V(t)\bigr)
    \le
    \frac{\gamma c_M}{\eta}V(0)$. Hence $\theta$ is also bounded, nondecreasing, and converges to a finite
limit.
\end{pf}

This result is not the main recovery
mechanism, since it requires the sign-mismatch terms to be dominated by the nominal Lyapunov
dissipation, but it shows that both schedulers preserve boundedness under arbitrary locally
finite sign variations. Furthermore, even if condition \eqref{eq:common_degraded_condition} is met, exact nominal stability properties are not recovered. 

For that purpose, we will now present the chosen schedulers starting by the vector case. We use the following terminology. On an interval \([T_a,T_b]\) where
$S(t)\equiv\bar S:=\diag(\bar s_1,\ldots,\bar s_m)$ is constant, a scheduler interval is called correct if the sign matrix
assigned to that interval satisfies \(\bar S N=I_m\). Thus, the wrapped controller coincides exactly with the
nominal closed loop. A correct interval is called trapping on \([T_c,T_b]\),
\(T_c\in[T_a,T_b]\), if, once entered at \(T_c\), the scheduler cannot leave it
before \(T_b\). Equivalently, \(\bar S N(\theta(t))=I_m\) for all \(t\in[T_c,T_b]\), and exact
nominal recovery holds throughout that interval.

\subsection{Independent vector scheduler}
\label{subsec:vector_scheduler}

The vector scheduler is defined by \(\theta=(\theta_1,\ldots,\theta_m)^\top\) and
\(N(\theta)=\diag(N_1(\theta_1),\ldots,N_m(\theta_m))\). Each channel
\(\theta_i\) evolves according to \eqref{eq:mimo_theta_i}. The switching
thresholds are chosen as
\begin{equation}
    \Theta_{i,0}=0,
    \qquad
    \Theta_{i,k+1}-\Theta_{i,k}=\gamma_i\lambda_k,
    \qquad
    \lambda_k=\lambda_0r^k,
    \label{eq:vector_thresholds}
\end{equation}
for \(i=1,\ldots,m\) and \(k\in\mathbb N_0\), where \(\lambda_0>0\) and
\(r>1\) are design parameters. For each \(\theta_i\), define
\(k_i:=\max\{k\in\mathbb N_0:\Theta_{i,k}\le\theta_i\}\). Then the
\(i\)-th sign compensation is
\begin{equation}
    N_i(\theta_i)=(-1)^{k_i}
    \quad\text{if}\quad
    \theta_i\in[\Theta_{i,k_i},\Theta_{i,k_i+1}),
    \label{eq:mimo_N_i}
\end{equation}
for \(i=1,\ldots,m\). Therefore, once the scheduler \eqref{eq:vector_thresholds} is a priori fixed, the evolution of the auxiliary variables $\theta_i$ according to \eqref{eq:mimo_theta_i} will produce the switching sign compensation \eqref{eq:mimo_N_i}. The next theorem presents the theoretical result.

\begin{thm}
\label{thm:vector_scheduler}
Consider the extended closed loop \eqref{eq:mimo_extended_closed_loop}
with the vector scheduler \eqref{eq:mimo_theta_i}, \eqref{eq:vector_thresholds},
\eqref{eq:mimo_N_i}, $\lambda_0>0$ and $r>1$, under
Assumptions~\ref{ass:mimo_nominal_lyap}--\ref{ass:mimo_Psi_bounds}. Let
$[T_a,T_b]$ be an interval on which
$S(t)\equiv\bar S:=\diag(\bar s_1,\ldots,\bar s_m)$ is constant. Assume that,
for some $T_c\in[T_a,T_b]$, the vector scheduler reaches a correct interval,
namely $\bar S N(\theta(T_c))=I_m$. Define $d_i(T_c):=
    \gamma_i^{-1} \left(\theta_i(T_c)-\Theta_{i,k_i} \right)$. If
\begin{equation}
    \min_{i=1,\ldots,m}
    \left\{
        \lambda_0 r^{k_i}
        -
        d_i(T_c)
    \right\}
    >c_MV(T_c),
    \label{eq:direct_trapping_condition_vector}
\end{equation}
then the correct interval is trapping. Hence $\bar S N(\theta(t))=I_m$ for all
$t\in[T_c,T_b]$, and the actual closed loop coincides with the nominal closed
loop on $[T_c,T_b]$. Moreover, the auxiliary variables remain bounded on
$[T_c,T_b]$. If $[T_a,T_b]=[T_a,\infty)$, then each $\theta_i(t)$ converges to a finite
limit and the original variables $(\sigma,\nu)$ converge to the origin with the
nominal stability property.
\end{thm}

\begin{pf}
Introduce the normalized variables $\vartheta_i:=\theta_i/\gamma_i$, so that
$\dot\vartheta_i=\psi_i$. While the solution remains in the correct interval,
\eqref{eq:mimo_correct_bound} is recovered. Using \eqref{eq:mimo_Psi}, \eqref{eq:mimo_theta_i}, \eqref{eq:mimo_Psi_bounds}, it follows that
$    \dot V
    \le
    -\frac{1}{c_M}\sum_{i=1}^m\dot\vartheta_i $. Hence, as long as the trajectory remains in the correct interval,
\begin{equation}
    \frac{d}{dt}
    \left(
        V+\frac{1}{c_M}\sum_{i=1}^m\vartheta_i
    \right)
    \le0 .
    \label{eq:vector_trapping_storage}
\end{equation}
Assume, by contradiction, that the trajectory leaves this interval at a first time
$T_e\in(T_c,T_b]$. Then at least one channel, say channel $j$, reaches its next
threshold. Therefore, $    \vartheta_j(T_e)-\vartheta_j(T_c)
    =
    \gamma_j^{-1} \left(\Theta_{j,k_j+1}-\theta_j(T_c) \right)$, is the normalized distance still available before the next switch of the $j$-th scheduler. Splitting it into the full normalized interval length
minus the part already traveled by the adaptive variable $\theta(t)$ at time $T_c$, one obtains
$    \gamma_j^{-1} \left(\Theta_{j,k_j+1}-\theta_j(T_c) \right)
    =
    \gamma_j^{-1} \left(\Theta_{j,k_j+1}-\Theta_{j,k_j} \right)
    -
    \gamma_j^{-1} \left(\theta_j(T_c)-\Theta_{j,k_j} \right)
    =
    \lambda_0 r^{k_j}
    -
    d_j(T_c)
    =
    \ell_j(T_c)$. Thus the distance required to leave this interval satisfies
$\vartheta_j(T_e)-\vartheta_j(T_c)=\ell_j(T_c)$.

Integrating \eqref{eq:vector_trapping_storage} on
$[T_c,T_e]$ gives
$    \sum_{i=1}^m
    \left(
        \vartheta_i(T_e)-\vartheta_i(T_c)
    \right)
    \le c_MV(T_c)$. Since all $\vartheta_i$ are nondecreasing, the last expression
implies
$    \ell_j(T_c)
    =
    \vartheta_j(T_e)-\vartheta_j(T_c)
    \le c_MV(T_c)$. Consequently,
$    \min_{i=1,\ldots,m}
    \left\{
        \lambda_0 r^{k_i}
        -
        d_i(T_c)
    \right\}
    \le c_MV(T_c)$, which contradicts \eqref{eq:direct_trapping_condition_vector}. Therefore the
correct interval cannot be escaped on $[T_c,T_b]$.

Since $\bar S N(\theta)=I_m$ on this interval and $B_0B_0^+=I_n$, the real
closed loop coincides with the nominal closed loop. Moreover,
$\theta_i(t)<\Theta_{i,k_i+1}$ for every $t\in[T_c,T_b]$, while
$\dot\theta_i\ge0$. Hence the auxiliary variables are bounded on
$[T_c,T_b]$. If $[T_a,T_b]=[T_a,\infty)$, monotonicity of the auxiliary variables $\theta_i$ implies that each $\theta_i(t)$ has a finite limit, and the convergence of $(\sigma,\nu)$ follows from the nominal
Lyapunov result \eqref{eq:mimo_nominal_lyap}.
\end{pf}

Condition \eqref{eq:direct_trapping_condition_vector} can be easily implemented numerically allowing the application of the Theorem.

\begin{cor}
\label{cor:vector_scheduler_online}
Consider the vector scheduler of Theorem~\ref{thm:vector_scheduler}.
Let $T_\ell$ be a switching instant and let $k_{i,\ell}$ be the
post-switch active index of the $i$th channel, namely
$\theta_i(T_\ell)\in[\Theta_{i,k_{i,\ell}},\Theta_{i,k_{i,\ell}+1})$.
Fix $\varepsilon>0$ and define
\begin{equation}
\label{eq:tail_shift}
\Delta_{i,\ell}
:=
\max\left\{
0,\,
\theta_i(T_\ell)+\gamma_i\big(c_MV(T_\ell)+\varepsilon\big)
-\Theta^-_{i,k_{i,\ell}+1}
\right\}.
\end{equation}
At $T_\ell$, update the scheduler thresholds by shifting the future intervals,
\begin{equation}
\label{eq:tail_update}
\Theta^+_{i,j}
=
\Theta^-_{i,j}
+
\Delta_{i,\ell},
\qquad
j\ge k_{i,\ell}+1,
\end{equation}
while keeping $\Theta^+_{i,j}=\Theta^-_{i,j}$ for $j\le k_{i,\ell}$.
Between switching instants, the scheduler thresholds are kept fixed.

Then the threshold ordering is preserved and the interval entered at $T_\ell$
satisfies
\begin{equation}
\label{eq:residual_length_online}
\min_{i=1,\ldots,m}
\frac{
\Theta^+_{i,k_{i,\ell}+1}-\theta_i(T_\ell)
}{\gamma_i}
\ge c_MV(T_\ell)+\varepsilon .
\end{equation}
Consequently, if this interval is correct for the current constant sign matrix,
namely $\bar S N(\theta(T_\ell^+))=I_m$, then it is trapping. Hence
$\bar S N(\theta(t))=I_m$ and the actual closed loop coincides with the nominal
closed loop on that interval. If this interval is unbounded, then each $\theta_i(t)$ converges to a finite
limit and the original variables $(\sigma,\nu)$ converge to the origin with the
nominal stability property.
\end{cor}

\begin{pf}
Since the same nonnegative quantity $\Delta_{i,\ell}$ is added to all
future thresholds of channel $i$, the ordering of the threshold sequence is
preserved. Moreover, by construction,
$\Theta^+_{i,k_{i,\ell}+1}
\ge
\theta_i(T_\ell)+\gamma_i\big(c_MV(T_\ell)+\varepsilon\big)$,
which gives \eqref{eq:residual_length_online}. Therefore, if the interval
entered at $T_\ell$ is correct, condition \eqref{eq:direct_trapping_condition_vector}
of Theorem~\ref{thm:vector_scheduler} is satisfied with margin
$\varepsilon$. The trapping property and the recovery of the nominal closed
loop follow directly from Theorem~\ref{thm:vector_scheduler}.
\end{pf}

\begin{rem}
The update \eqref{eq:tail_shift}--\eqref{eq:tail_update} is sign-blind: if the visited interval is wrong, its residual length is also enlarged. This could increase the transient spent with an incorrect sign. Therefore, the update should be interpreted as a recovery-oriented implementation rule.
\end{rem}

\subsection{Scalar pattern scheduler}
\label{subsec:scalar_scheduler}

The scalar scheduler uses one variable to visit complete input-sign patterns.
Let $\mathscr S_m=\{\Pi_j\}_{j=1}^{2^m}$ be an enumeration of the diagonal sign
matrices
\begin{equation}
    \Pi_j=\diag(\pi_{j1},\ldots,\pi_{jm}),
    \qquad
    \pi_{ji}\in\{-1,+1\}.
    \label{eq:scalar_pattern_set}
\end{equation}
For instance, for $m=2$ one may choose
\begin{equation}
    \Pi_1=
    I_2,
    \quad
    \Pi_2=
    \begin{bmatrix}
        -1 & 0\\
        0 & 1
    \end{bmatrix},
    \quad
    \Pi_3=
    \begin{bmatrix}
        1 & 0\\
        0 & -1
    \end{bmatrix},
    \quad
    \Pi_4=
    -I_2.
\label{eq:scalar_pattern_set_m2}
\end{equation}
The particular ordering is unimportant, provided that all sign patterns $\Pi_j$ are visited by the scheduler.

The auxiliary variable $\theta$ satisfies the scalar adaptation \eqref{eq:scalar_theta}. Define the scalar scheduler
\begin{equation}
    \Theta_0=0,
    \qquad
    \Theta_{k+1}-\Theta_k=\gamma\lambda_k 
    \qquad
    \lambda_k=\lambda_0r^k,
    \label{eq:scalar_thresholds}
\end{equation}
where $\lambda_0>0$ and $r>1$ are parameters to be selected.
Let $\mu(k):=1+(k\bmod 2^m)$. The sign-compensation matrix in
\eqref{eq:mimo_adaptive_control} is
\begin{equation}
    N(\theta)=\Pi_{\mu(k)}
    \quad
    \text{if }
    \theta\in[\Theta_k,\Theta_{k+1}).
    \label{eq:scalar_pattern_N}
\end{equation}

Therefore, once the scheduler \eqref{eq:scalar_pattern_set}, \eqref{eq:scalar_thresholds} is a priori fixed, the evolution of the auxiliary variable $\theta$ according to \eqref{eq:scalar_theta} will produce the switching sign compensation \eqref{eq:scalar_pattern_N}. The next theorem presents the theoretical result.

\begin{thm}
\label{thm:scalar_scheduler}
Consider the extended closed loop \eqref{eq:mimo_extended_closed_loop}
with the scalar scheduler \eqref{eq:scalar_theta}, \eqref{eq:scalar_thresholds},
\eqref{eq:scalar_pattern_N}, $\lambda_0>0$ and $r>1$, under
Assumptions~\ref{ass:mimo_nominal_lyap}--\ref{ass:mimo_Psi_bounds}. Let
$[T_a,T_b]$ be an interval on which $S(t)\equiv\bar S\in\mathscr S_m$ is
constant. If
\begin{equation}
    r>1+2c_M,
    \label{eq:scalar_r_condition}
\end{equation}
then every sufficiently large correct pattern interval contained in
$[T_a,T_b]$ is trapping. Hence, if the scalar scheduler reaches such an interval
at some $T_k\in[T_a,T_b]$, then $\bar S N(\theta(t))=I_m$ for all
$t\in[T_k,T_b]$, and the actual closed loop coincides with the nominal closed
loop on $[T_k,T_b]$. Moreover, the auxiliary variable remains bounded on
$[T_k,T_b]$. If $[T_a,T_b]=[T_a,\infty)$, then the scalar scheduler is unconditional in the following sense: either a sufficiently large correct interval is reached and trapped in finite time, or $\theta$ remains bounded and $(\sigma,\nu)$ converges to the origin before further sign exploration is needed. In all cases, $\theta(t)$ converges to a finite limit and the control objective is achieved.
\end{thm}

\begin{pf}
Introduce the normalized variable $\vartheta:=\theta/\gamma$, so that
$\dot\vartheta=\Psi$. Let $[\Theta_k,\Theta_{k+1})$ be a correct pattern
interval entered at time $T_k\in[T_a,T_b]$, namely
$\theta(T_k)=\Theta_k$ and $\Pi_{\mu(k)}=\bar S$. 

Let $\Lambda_k:=\Theta_k/\gamma$ be the normalized scheduler thresholds. Then
$\Lambda_{k+1}-\Lambda_k=\lambda_k$, with $\lambda_k=\lambda_0r^k$.
Let $h$ be such that
$\vartheta(T_a)\in[\Lambda_h,\Lambda_{h+1})$. Since $\vartheta$ is
nondecreasing and the scheduler enters the interval
$[\Lambda_k,\Lambda_{k+1})$ at $T_k\ge T_a$, one has $h\le k$ and
$    \vartheta(T_k)-\vartheta(T_a)
    =
    \Lambda_k-\vartheta(T_a)
    \le
    \Lambda_k-\Lambda_h
    =
    \sum_{\ell=h}^{k-1}\lambda_\ell $. As the normalized interval lengths are exponential,
$   \sum_{\ell=h}^{k-1}\lambda_\ell
    =
    \lambda_0\sum_{\ell=h}^{k-1}r^\ell
    =
    \lambda_0\frac{r^k-r^h}{r-1}
    <
    \frac{\lambda_0r^k}{r-1}
    =
    \frac{\lambda_k}{r-1}$.
Therefore,
$    \vartheta(T_k)-\vartheta(T_a)
    <
    \frac{\lambda_k}{r-1}$. From \eqref{eq:mimo_wrong_bound}, $\dot V\le -W+2\Psi\le2\dot\vartheta$ on
$[T_a,T_k]$. Thus,
$ V(T_k)
    \le
    V(T_a)+\frac{2}{r-1}\lambda_k $.

While $\theta(t)\in[\Theta_k,\Theta_{k+1})$, one has $N(\theta)=\bar S$ and
therefore $\bar S N(\theta)=I_m$. Hence, \eqref{eq:mimo_Psi_bounds}, \eqref{eq:mimo_correct_bound} give
$    \dot V
    \le
    -\frac{1}{c_M}\dot\vartheta $. Consequently, as long as the trajectory remains in the correct interval,
\begin{equation}
    \frac{d}{dt}
    \left(
        V+\frac{1}{c_M}\vartheta
    \right)
    \le0 .
    \label{eq:scalar_monotone_quantity}
\end{equation}
If the trajectory exits the correct interval, then $\lambda_k \le c_M V(T_k)
\le c_M V(T_a)+\frac{2c_M}{r-1}\lambda_k $. Equivalently, $\left(1-\frac{2c_M}{r-1}\right)\lambda_k
\le c_M V(T_a)$. Since $r>1+2c_M$, the coefficient on the left-hand side is positive, while
$\lambda_k=\lambda_0 r^k\to\infty$. Hence the inequality is impossible for all
sufficiently large correct intervals. Therefore, every
sufficiently large correct pattern interval is trapping.

Once such an interval is trapping, $\theta(t)<\Theta_{k+1}$ and
$\dot\theta\ge0$ for all $t\in[T_k,T_b]$. Hence $\theta$ remains bounded on
$[T_k,T_b]$. Since $\bar S N(\theta)=I_m$ on this interval and
$B_0B_0^+=I_n$, the real closed loop coincides with the nominal one.

Assume now that \(\theta\) is bounded on $[T_a,T_b]=[T_a,\infty)$. Since \(\theta\) is nondecreasing, \(\theta(t)\) has a finite limit. Hence, integrating the scheduler dynamics gives \(\int_{T_a}^{\infty}\Psi(t)dt<\infty\). From \eqref{eq:mimo_wrong_bound}, it follows that \(\int_{T_a}^{\infty}W(\sigma(t),\nu(t))dt<\infty\). Since the solution is bounded, convergence follows by Barbalat's lemma in the piecewise-Carath\'eodory case, or by the corresponding weak invariance argument in the Filippov case (see Assumption \ref{ass:smooth}). This alternative yields convergence of \((\sigma,\nu)\) without necessarily identifying the correct sign pattern. Since the matrices in $\mathscr S_m$ are assigned periodically, if $\theta$ is unbounded, every sign pattern is visited infinitely many times; hence a sufficiently large correct interval is reached. By the first part of the proof, this interval is trapping, so $\theta$ becomes bounded, which is a contradiction. Thus $\theta(t)$ always converges to a finite limit. Moreover, whenever convergence does not occur before exploration stops, the scalar scheduler must recover a correct interval in finite time, after which the nominal Lyapunov result \eqref{eq:mimo_nominal_lyap} applies.
\end{pf}

\subsection{Comparison between the two schedulers}
\label{subsec:scheduler_comparison}

Both schedulers implement the same modular compensation law
\eqref{eq:mimo_adaptive_control}; hence, neither modifies the nominal control
magnitude. They differ in the tradeoff between scalability and pattern
reachability.

The vector scheduler uses one auxiliary variable per input channel, so its
complexity grows linearly with \(m\). Its result is conditional: if a correct
vector cell is reached with sufficient residual length, then
Theorem~\ref{thm:vector_scheduler} makes it trapping. The online update of
Corollary~\ref{cor:vector_scheduler_online} enforces this residual-length
condition at every visited cell, so recovery is certified whenever the visited
cell is correct. The remaining reachability question is application-dependent:
it can be checked along the closed-loop trajectory, or strengthened in specific
plants using additional excitation, phase-selection, or supervisory exploration
rules. Developing such coordinated vector schedulers while preserving linear or
near-linear complexity is left for future work.

The scalar scheduler uses a single auxiliary variable to enumerate all matrices
in \(\mathscr S_m\). Thus, on a constant-sign interval, correct patterns occur
periodically and their lengths grow unbounded. Consequently, under
\(r>1+2c_M\), Theorem~\ref{thm:scalar_scheduler} guarantees that either a
sufficiently large correct interval is reached and trapped, or the state
converges before further exploration is needed. The price is combinatorial,
since \(|\mathscr S_m|=2^m\). In the extreme and physically unrealistic case where the plant sign changes are coincident with the scheduler's so that \(S(t)\) changes whenever a correct pattern is
entered, no nontrivial constant-sign interval containing a correct scheduler
interval is available. Consequently, Theorem~\ref{thm:scalar_scheduler} does not
apply, and the existence of a trapping interval cannot be guaranteed. 

For SISO systems, the two constructions coincide.

\section{Numerical validation}
\label{sec:Sims}

All simulations were performed in Python with \texttt{solve\_ivp} function using `BDF' method \cite{b:doi:10.1137/S1064827594276424}. Sign changes were handled by event-based restarts: whenever a prescribed sign-switching time or scheduler threshold was reached, \(S(t)\) and/or \(N(\theta)\) was updated and the integration resumed from the same state.

\subsection{Flight vehicle roll reversal}
\label{sec:flight_roll_reversal}

We consider a flight-vehicle roll channel affected by the roll-reversal phenomenon. This effect appears in canard-controlled or highly maneuverable flight vehicles when the aerodynamic rolling moment generated by a control-surface deflection changes sign as a function of the angle of attack, sideslip angle, or flight velocity. In this situation, a command intended to produce a positive roll acceleration may produce the opposite effect, as illustrated in Fig.~\ref{fig:airplane}. Therefore, the control direction is not only unknown, but may also change during the maneuver.

\begin{figure}[ht!]
\centering
\includegraphics[width=7.0cm,keepaspectratio]{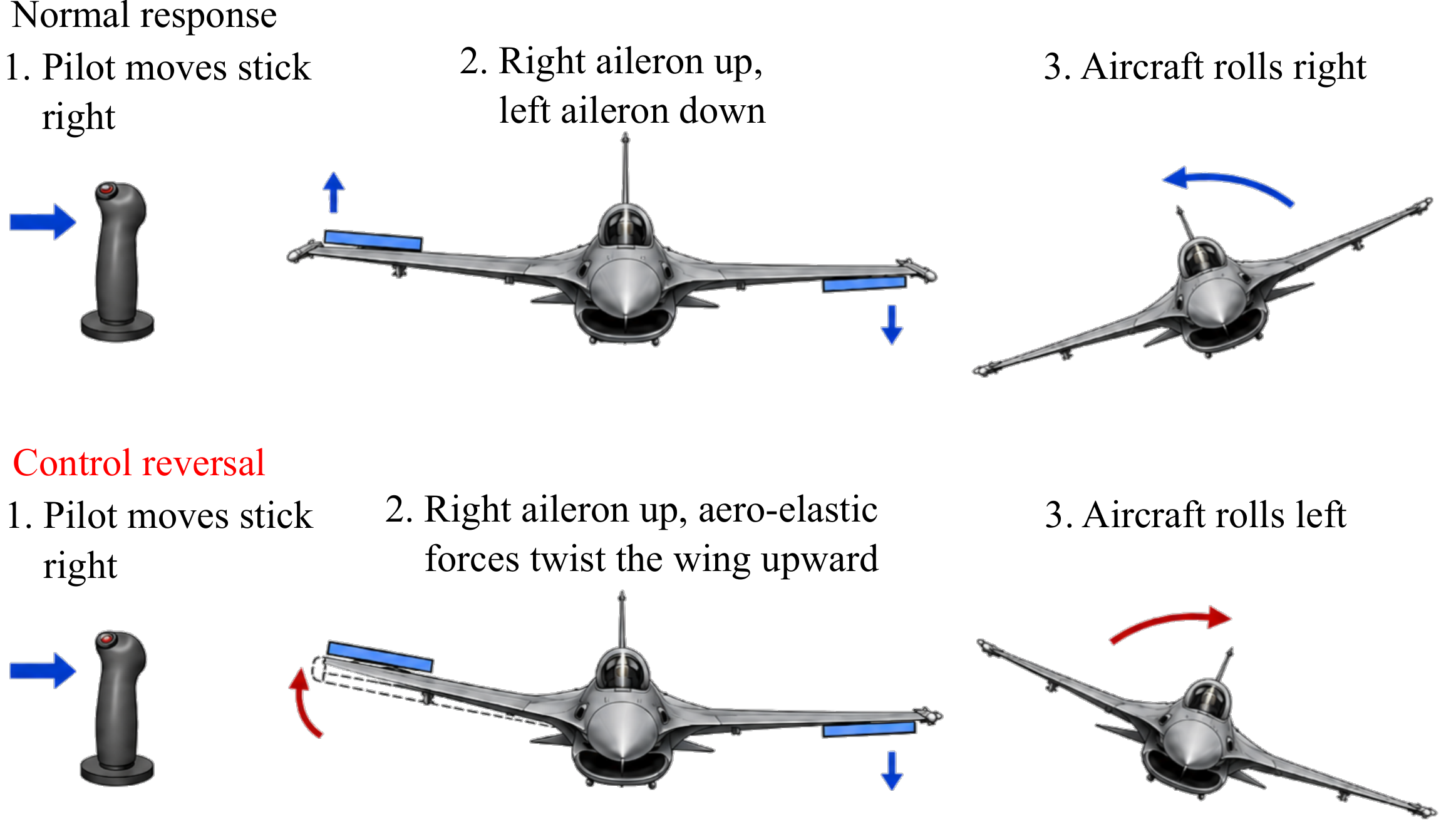}
\caption{Schematic representation of the roll-reversal mechanism. At low angle of attack (or speed), the control-surface deflection produces the nominal roll direction, whereas at high angle of attack (or speed) the induced aerodynamic interaction may reverse the effective roll moment.}
\label{fig:airplane}
\end{figure}

Following the roll-reversal model in~\cite{b:mirzaei2018roll}, we focus on the reduced roll dynamics
\begin{equation}
\dot\phi=\omega,\qquad
\dot \omega=L_p \omega+L_{\delta_a}(t)\delta_a+d(t),
\label{eq:flight_roll_model}
\end{equation}
where $\phi$ is the roll angle, $\omega$ is the roll rate, $\delta_a$ is the aileron command, $L_p<0$ is the roll damping coefficient, $d(t)$ is an external roll disturbance, and $L_{\delta_a}(t)$ is the sign-changing roll-control effectiveness. The control objective is to stabilize $\phi$ and $\omega$ despite the unknown sign of $L_{\delta_a}(t)$. 

Following~\cite{b:mirzaei2018roll}, define the sliding variable $\sigma=\omega+C\phi$, $C>0$, and the nominal Lyapunov function $V=\frac{1}{2}\sigma^2$. If the sign of $L_{\delta_a}(t)$ is known and $|L_{\delta_a}(t)|=\hat L_{\delta_a}$, the \textit{baseline} controller can be written as
\begin{equation}
\delta_a=S(t)\bar\delta_a, \quad 
\bar\delta_a
=
-\frac{1}{\hat L_{\delta_a}}
\left(
L_p \omega+C \omega+\rho_1
\frac{\sigma}{|\sigma|+\varepsilon}
\right),
\label{eq:flight_nominal_delta}
\end{equation}
where $S(t)=\operatorname{sign}(L_{\delta_a}(t))$ and $\rho_1>0$. 

We implement our modular switching sign compensator \eqref{eq:mimo_p_q_def}, \eqref{eq:mimo_adaptive_control} using the nominal control $\bar\delta_a$ in \eqref{eq:flight_nominal_delta}. Since these error dynamics are SISO, the vector and scalar schedulers presented in Section \ref{sec:Control} are the same, so we choose \eqref{eq:scalar_theta}, \eqref{eq:vector_thresholds}, \eqref{eq:mimo_N_i} for the compensation of the unknown sign.

For comparison, we also implement the online sign identifier proposed in~\cite{b:mirzaei2018roll}. According to~\cite{b:mirzaei2018roll}, on each interval $[t_1,t_2]$, the sign of the input gain is estimated as 
$\hat S(t_2)
=
\operatorname{sign}
\left[
\omega(t_2)-\omega(t_1)-L_p^{\rm id}
(\phi(t_2)-\phi(t_1))
\right]$ $\times 
\operatorname{sign}
\left(
\int_{t_1}^{t_2}\delta_a(t) dt
\right)$.
The corresponding controller is $\delta_a=\hat S(t)\bar\delta_a$. This estimator is very effective when the right-hand side of \eqref{eq:flight_roll_model} is exactly known. However, its performance is sensitive to model mismatch, because it explicitly uses the roll model to reconstruct the sign of $L_{\delta_a}(t)$. To test this limitation, the identifier uses a nominal value $L_p^{\rm id}$ instead of the real value $L_p$.

The numerical values were $C=2$, $\rho_1=5$, $\varepsilon=0.1$, $L_p=-0.1$, $\hat L_{\delta_a}=0.2$, and $\delta_a$ was saturated at $90$ deg. The true changing roll-control effectiveness, $L_{\delta_a}(t)$, was selected as a square wave with amplitude $0.2$ and period $5$ s. The matched disturbance $d(t)=\sin(0.01t)$ is covered by the nominal robust term in the Filippov sense; the regularization $\varepsilon$ is only its numerical implementation. The initial condition was $\phi(0)=-20$ deg and $\omega(0)=20$ deg/s. The identifier was updated every $0.5$ s and used $L_p^{\rm id}=-0.3$. The scalar scheduler used $\theta(0)=0$, $\gamma_\theta=0.9$, $\lambda_0=268$, and $r=1.005$.

The results are reported in Table~\ref{tab:combined_examples} and shown in Fig.~\ref{fig:flight_results}. The baseline gives the ideal response and is included only to indicate the performance attainable with exact sign knowledge. Under the model mismatch in the identifier, the method of~\cite{b:mirzaei2018roll} produces several unnecessary switches and does not recover the ideal sign sequence. This leads to a large residual roll angle, a high MISE of the sliding variable, and a larger RMS control effort. In contrast, the proposed scalar scheduler performs only three switches, synchronized with the sign changes of $L_{\delta_a}(t)$, and recovers a final accuracy closer to the baseline. Therefore, this example illustrates the effectiveness and robustness of the proposed approach.

\begin{table}[ht!]
\centering
\caption{Comparison: roll-reversal and visual-servoing examples.}
\label{tab:combined_examples}
\scriptsize
\setlength{\tabcolsep}{2.2pt}
\begin{tabular}{lcccc}
\hline
\multicolumn{5}{c}{Flight roll reversal}\\
\hline
Method & Max. $|\sigma|$ & MISE & RMS $\delta_a$ & Switches\\
\hline
Baseline Control & $2.00{\rm e}{1}$ & $1.53{\rm e}{1}$ & $1.61{\rm e}{1}$ & $0$\\
Controller \cite{b:mirzaei2018roll} & $3.77{\rm e}{1}$ & $2.17{\rm e}{2}$ & $3.96{\rm e}{1}$ & $11$\\
Proposed Control & $2.21{\rm e}{1}$ & $7.14{\rm e}{1}$ & $3.48{\rm e}{1}$ & $3$\\
\hline
\multicolumn{5}{c}{Robotic visual servoing}\\
\hline
Method & Max. $\|\sigma\|$ & MISE & RMS $u$ & Switches\\
\hline
Baseline Control & $1.94{\rm e}{-5}$ & $5.12{\rm e}{-12}$ & $3.25$ & $0$\\
Controller \cite{b:oliveira2010sliding} & $1.33$ & $2.28{\rm e}{-1}$ & $4.19$ & $3$\\
Proposed Control & $1.69$ & $1.93{\rm e}{-1}$ & $3.47$ & $3$\\
\hline
\end{tabular}
\end{table}

\begin{figure}[t!]
\centering
\includegraphics[width=0.80\columnwidth,keepaspectratio]{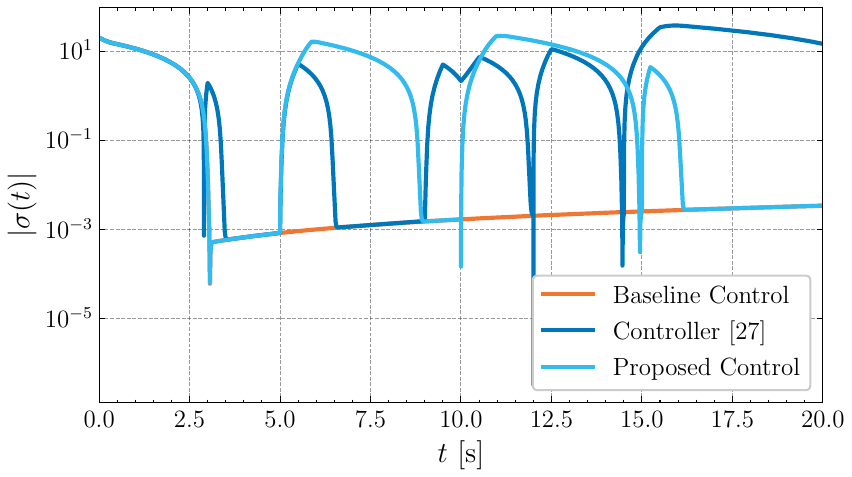}\\[-2mm]
{\small (a) Sliding variable}\\

\includegraphics[width=0.80\columnwidth,keepaspectratio]{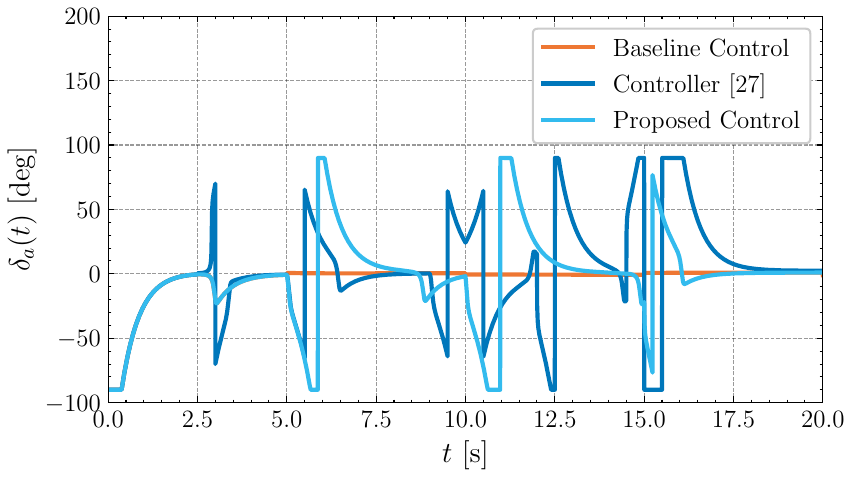}\\[-2mm]
{\small (b) Aileron command}\\

\includegraphics[width=0.80\columnwidth,keepaspectratio]{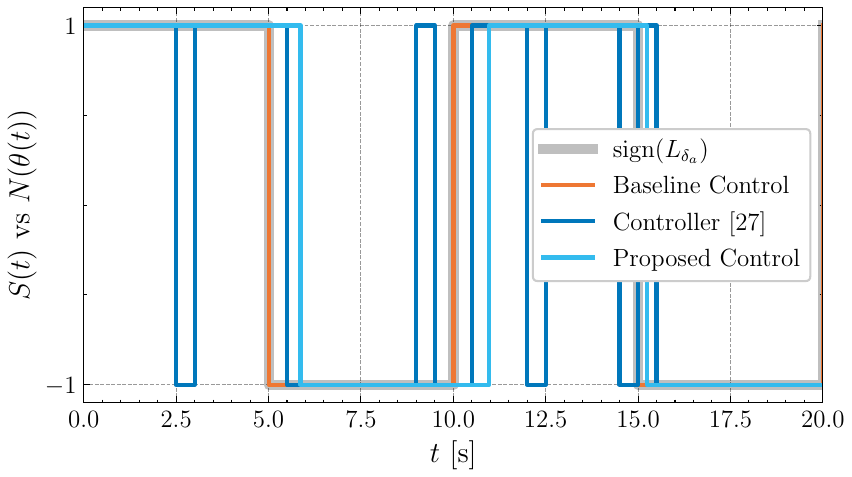}\\[-2mm]
{\small (c) Sign compensation}

\caption{Roll-reversal example. (a) Evolution of the sliding variable
$|\sigma(t)|$. (b) Aileron command $\delta_a(t)$. (c) Active sign
compensator. Under parameter mismatch, the model-based identifier switches
repeatedly, whereas the scalar scheduler performs only the required sign
changes and gives a smaller MISE.}
\label{fig:flight_results}
\end{figure}

\subsection{Robotic visual servoing}
\label{sec:robotic}

We consider the visual servoing benchmark of a planar robotic manipulator observed by a fixed uncalibrated camera. Let $y\in\mathbb{R}^2$ denote the position of the image feature attached to the end-effector, expressed in pixel coordinates, and let $y_m\in\mathbb{R}^2$ be the desired image trajectory. The image-plane dynamics are written as (see \cite{b:oliveira2010sliding,b:oliveira2011output,b:951370} for more details)
\begin{equation}
    \dot y = K_p(\psi)u+\phi_1(y),
    \qquad
    \phi_1(y)=
    \begin{bmatrix}
        y_1^2, \quad y_2^2
    \end{bmatrix}^T,
    \label{eq:visual-servoing-y}
\end{equation}
where $u\in\mathbb{R}^2$ is the Cartesian velocity command and $K_p(\psi)$ is the high-frequency gain matrix induced by the camera geometry. In particular,
\begin{equation}
    K_p(\psi)=
    \frac{f_0}{z_0+f_0}
    \begin{bmatrix}
        h_1&0\\
        0&h_2
    \end{bmatrix}
    R(\psi),
    \label{eq:visual-servoing-Kp}
\end{equation}
where $f_0$ is the focal length, $z_0$ is the average depth, $h_1,h_2$ are the camera scaling factors, and $R(\psi)$ is the planar rotation matrix associated with the unknown camera misalignment angle $\psi$. The reference is generated by $\dot y_m=-y_m+r(t)$, $r(t)= [2\cos t, 2\sin t]^T$. Defining the tracking error $\sigma=y-y_m$, one obtains
$    \dot\sigma
    =
    K_p(\psi)u+\phi(\sigma,t)$, with $\phi(\sigma,t)
    =
    \phi_1(\sigma+y_m)+y_m-r(t)$. The control objective is to steer $\sigma$ to the origin despite the unknown control direction induced by the uncalibrated camera.

We implement the unit-vector control law used in the monitoring-function approach \cite{b:1254099}. When the correct pre-compensator $S_\star$ is known, the baseline controller is
\begin{equation}
    u
    =
    S_\star \bar{u}, \qquad
    \bar{u} = -\rho(\sigma,t)
    \frac{\sigma}{\sqrt{\|\sigma\|^2+\varepsilon^2}},
    \label{eq:visual-servoing-nominal}
\end{equation}
where $\varepsilon>0$ is used only for numerical regularization. The modulation function $\rho$ is chosen in the spirit of the unit-vector model-reference adaptive control (UV-MRAC \cite{b:1254099}) design as 
$    \rho(\sigma,t)
    =
    \delta+c_2\|\sigma\|+c_3\hat d(t)$, where $\hat d(t)$ is an implementable bound for the equivalent perturbation $\phi(\sigma,t)$. This unit-vector term is the nominal robust compensation that makes zero tracking error an equilibrium in the ideal nonsmooth model. The baseline case is included only as an ideal benchmark, since it assumes the stabilizing matrix $S_\star$ is known.

We implement our modular switching sign compensator \eqref{eq:mimo_p_q_def}, \eqref{eq:mimo_adaptive_control} using the nominal control $\bar{u}$ in \eqref{eq:visual-servoing-nominal}, and the scalar scheduler \eqref{eq:scalar_theta}, \eqref{eq:scalar_pattern_set}, \eqref{eq:scalar_thresholds}, \eqref{eq:scalar_pattern_N}. $\Pi_j$ was chosen as in \eqref{eq:scalar_pattern_set_m2} and $\mu(k)=1+k\operatorname{mod}4$.

For comparison, the monitoring-function controller \cite{b:oliveira2010sliding} is also implemented. It uses the same unit-vector law \eqref{eq:visual-servoing-nominal}, but replaces $S_\star$ by a candidate matrix $S_q$ selected from the finite set \eqref{eq:scalar_pattern_set_m2}. The switch is triggered when the tracking error reaches the monitoring envelope
$    \varphi_q(t)
    =
    \|\sigma(t_q)\|e^{-\lambda_m(t-t_q)}
    +(q+1)e^{-\lambda_c t}$. Thus, the monitoring mechanism changes the pre-compensator $S_q$ whenever $\|\sigma(t)\|=\varphi_q(t)$.

The numerical parameters were $f_0=6$ mm, $z_0=1$ m, $h_1=119$ pixel/mm, $h_2=102$ pixel/mm, $y_1(0)=y_2(0)=0$, $\psi=\pi$, and nominal angle $\psi_0=\pi/2$. Thus the camera uncertainty reduces to the diagonal sign case $S_\star=-I_2$ considered in the theory; general rotations are outside the diagonal-sign framework. The modulation function was implemented with $\delta=0.1$, $c_2=0.1$, and $c_3=1$, as reported in \cite{b:oliveira2010sliding}. The monitoring parameters were $\lambda_m=\lambda_c=0.9$, while the scalar scheduler used $\theta(0)=0$, $\gamma=50$, $\lambda_0=1$, and $r=13$. As both switching controllers use the same compensator pattern \eqref{eq:scalar_pattern_set_m2}, both have to do three switches in order to get to the baseline pre-compensator $S_\star=-I_2$ corresponding to $R(\psi=\pi)$.

The results are summarized in Table~\ref{tab:combined_examples} and depicted in Fig. \ref{fig:robot_results}. The baseline control gives the expected ideal response, with no switching and negligible tracking error. Both adaptive switching strategies identify the stabilizing matrix after three switches. However, the scalar scheduler reaches the stabilizing candidate much earlier than the monitoring function, with lower mean integrated squared error (MISE), and also reducing the RMS control effort. The price paid is a larger peak transient error (max $\norm{\sigma}$) caused by the faster open-loop exploration of the candidate set. Thus, in this example, our approach improves the tracking and control-energy metrics while preserving the same number of switches and the same final tracking accuracy.

\begin{figure}[ht!]
  \centering
  \includegraphics[width=6.5cm,keepaspectratio]{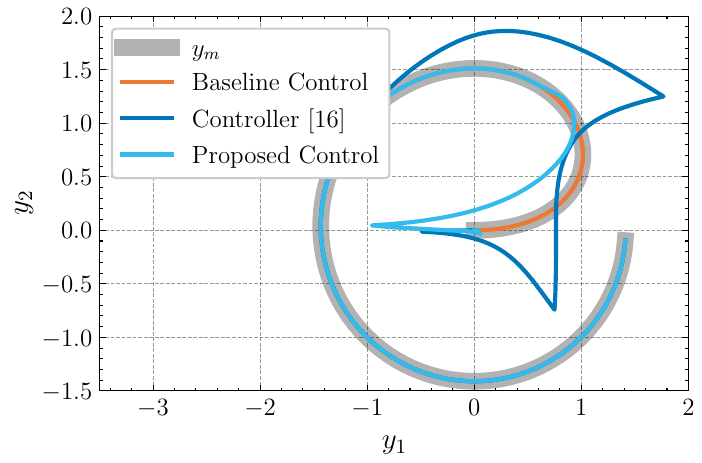}
  \caption{Image-plane trajectories for the robotic visual-servoing example. The gray curve represents the desired trajectory $y_m$, while the colored curves compare the baseline controller, the monitoring-function approach, and the proposed scalar scheduler.}
  \label{fig:robot_plane}
\end{figure}

\begin{figure}[t!]
\centering
\includegraphics[width=0.80\columnwidth,keepaspectratio]{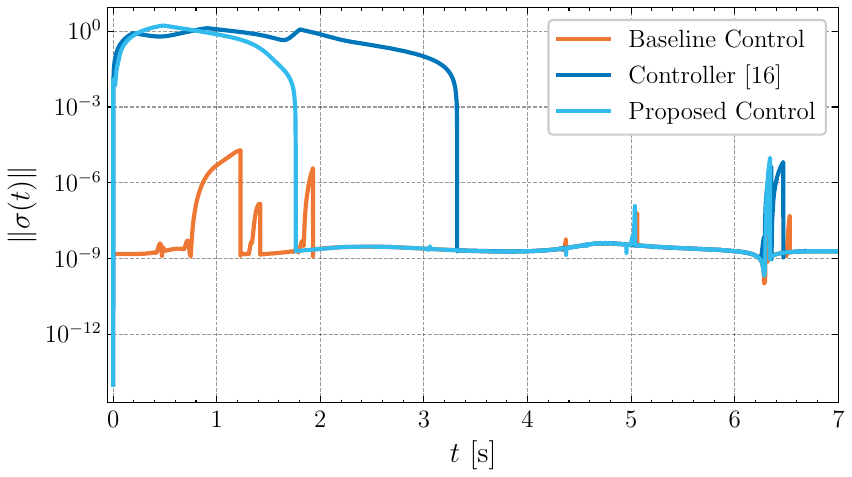}\\[-2mm]
{\small (a) Tracking error}\\

\includegraphics[width=0.80\columnwidth,keepaspectratio]{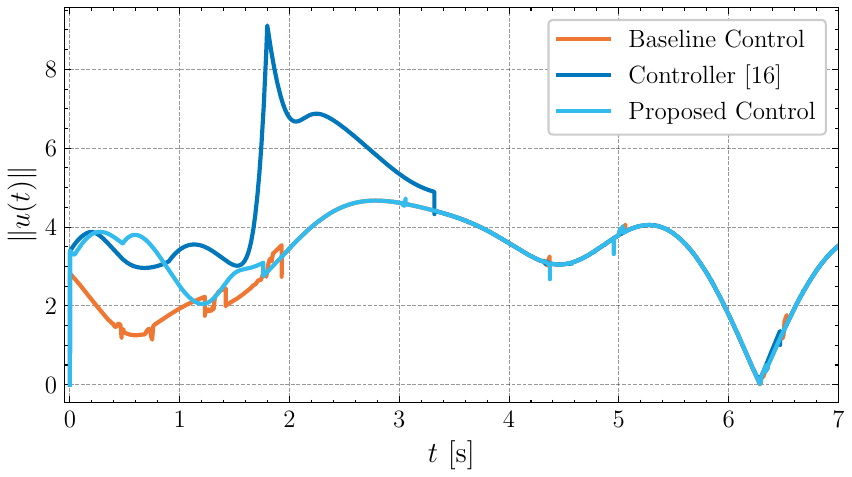}\\[-2mm]
{\small (b) Control effort}\\

\includegraphics[width=0.80\columnwidth,keepaspectratio]{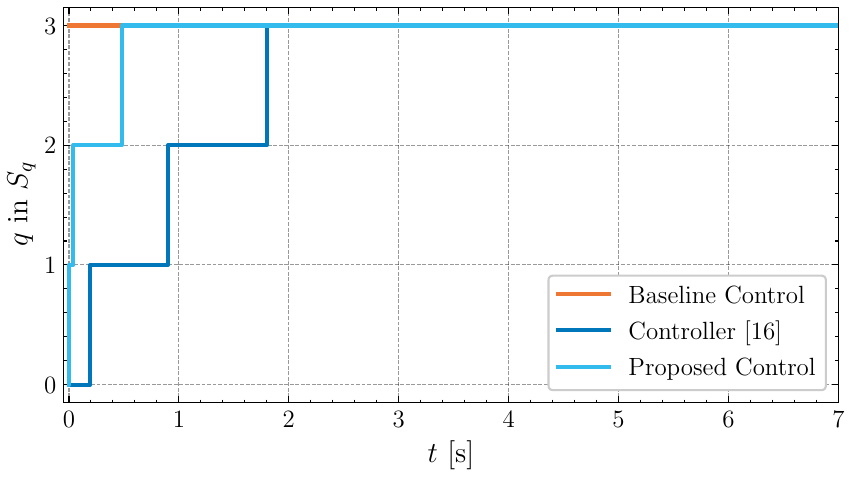}\\[-2mm]
{\small (c) Scheduler index}

\caption{Robotic visual-servoing example. (a) Evolution of the tracking error
norm $\|\sigma(t)\|$. (b) Control norm $\|u(t)\|$. (c) Active scheduler index
$q\in S_q$. The monitoring-function approach and the proposed scalar scheduler
switch among candidate control directions until the stabilizing sign
configuration is selected.}
\label{fig:robot_results}
\end{figure}

\subsection{Underground reservoir}
\label{sec:underground}

We now consider an academic underground reservoir example motivated by the need for
control of induced seismicity. The reservoir is described by a diffusion equation modelling fluid pressure changes $P(x,t)$, coupled with a point-wise logistic-like ODE representing the seismicity-rate density $R(x,t)$ (SR, see \cite{b:SMITH2022117697,b:Gutierrez-Stefanou-2025,b:kim2023,b:KIM2025103396,https://doi.org/10.1029/2018WR023587} for the application of this model in different reservoirs):
\begin{equation}
\begin{aligned}
    P_t(x,t) &= c_{\rm hy}\Delta P(x,t)
    + \frac{1}{\beta d_z}\sum_{j\in\{s,c_1,c_2\}} B_j(x)Q_j(t),\\
    R_t(x,t) &= R(x,t)
    \left[\gamma_1 P_t(x,t)-\gamma_2\big(R(x,t)-R^*\big)\right].
\end{aligned}
\label{eq:under_model}
\end{equation}
$x \in V$ is the spatial variable and $t \geq 0$ the time variable. The SR expresses the number of seismic events per unit time and volume. $Q_s$ denotes a well injecting fluid as shown in Fig. \ref{fig:under_results}. $Q_{c_1},Q_{c_2}$ are the controlled
wells. The function $B_j(x)$ is a smooth version of a Dirac distribution, \textit{i.e.}, the injection/extraction of fluid $Q_j(t)$ is applied over a given region $V_j^* \subset V$ instead of a point (see \cite{b:Gutierrez-Stefanou-2025} for more details on the model). The domain is $\Omega=[-2,2]\times[-2,2]$ km$^2$, with thickness
$d_z=0.2$ km as shown in Fig.~\ref{fig:reservoir}. The parameters used in the simulation are
$c_{\rm hy}=4.4\times 10^{-2}$ [km$^2$/h],
$\beta=5.7\times10^{-4}$ [MPa$^{-1}$],
$R^*=1$ [event year$^{-1}$ km$^{-3}$],
$f=0.5$ [-], $a=0.003$ [-],
$\sigma_n=16.67$ [MPa], $\dot\tau_0=10^{-7}$ [MPa/h],
$t_a=a\sigma_n/\dot\tau_0$ [h], $\gamma_1=f/(\dot\tau_0t_a)$ [MPa$^{-1}$], and
$\gamma_2=1/(t_a R^*)$ [km$^3$ events$^{-1}$]. The initial conditions were selected as $P(x,0)=0$, $R(x,0)=R^*$, and undrained boundary conditions were considered, \textit{i.e.}, $\nabla P(x,t) \cdot \hat{e}=0$ at $\partial V$, where $\hat{e}$ is a unit normal vector to $\partial V$. 
\begin{figure}[ht!]
  \centering
  \includegraphics[width=7.0cm,keepaspectratio]{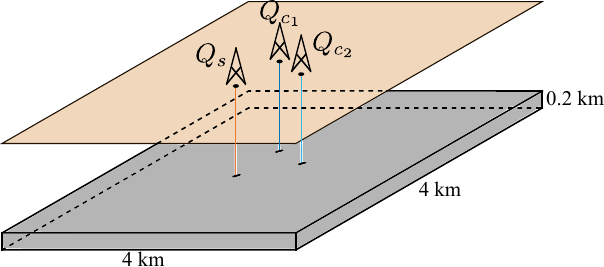}
  \caption{Schematic figure of the underground reservoir.}
  \label{fig:reservoir}
\end{figure}

\vspace{-10pt}
The control objective is to regulate the averaged SR over the whole domain ($y(t) = \frac{1}{V} \int_{V} R(x,t) \, dV$) output despite the fact that the physical mechanism producing seismicity is not assumed to be known a priori. Indeed, depending on various geophysical processes that can be hard to identify in advance, it is not always possible to know beforehand
whether injection or extraction is the stabilizing action (\textit{e.g.}, \cite{b:NAM2016GPM}). This uncertainty is
represented here by an unknown switching control direction
$S(t)=\diag(s_1(t),s_2(t))$, with $s_i(t)\in\{-1,1\}$, acting on the two controlled
channels.

We use the nominal stabilizing controller designed for the case $S(t)\equiv I_2$.
Let $\sigma(t) = \frac{1}{\gamma_{1_0} R^*_0}\left[y(t)-y_r(t) \right]$ be the regulated output error, where $\gamma_{1_0}, R^*_0$ are parameters to be selected, and $y_r(t)$ is the reference. Let $B_0$ be the nominal input matrix. The nominal flux control is selected according to \cite{b:Gutierrez-Stefanou-2025} as
\begin{equation}
\begin{aligned}
    \bar Q_c(t) &= B_0^+
    \left[-k_1\lceil \sigma(t)\rfloor^{\frac{1}{1-l}}+\nu(t)\right],\\
    \dot\nu(t) &= -k_2\lceil \sigma(t)\rfloor^{\frac{1+l}{1-l}},
\end{aligned}
\label{eq:under_nominal}
\end{equation}
where $\bar Q_c=[\bar Q_{c_1},\bar Q_{c_2}]^\top$, $k_1,k_2$ are positive gains, and $\lceil z\rfloor^\alpha:=|z|^\alpha\sign(z)$ is applied
component-wise. This is known as a MIMO super-twisting-type control \cite{b:Mathey-Moreno-2024}.

Since the direction of the stabilizing mechanism is unknown, the implemented control is
obtained by wrapping the nominal controller with the exponential vector scheduler introduced in \eqref{eq:mimo_theta_i}, \eqref{eq:vector_thresholds}, \eqref{eq:mimo_N_i}, with the online
threshold update \eqref{eq:tail_shift}, \eqref{eq:tail_update}. The control parameters were selected as $k_1=6.7 \times 10^{-3}$, $k_2=2.2 \times 10^{-5}$, $l=-1$, $\gamma_{i}=1 \times 10 ^{3}$, $R^*_0 = y_r(t) = R^*$, $\lambda_0=2 \times 10^{-4}$, $r=45$, $\varepsilon=1\times 10^{-8}$, and $\theta_i(0)=0$. The value $c_M=1\times 10^{-2}$ used in the online update was estimated from the simulated nominal Lyapunov ratio $\Psi/W$. To test our switching control, the unknown switching $S(t)$ changes every 2 [months], as seen in Fig. \ref{fig:under_results} (c).

The closed-loop behaviour is reported in Fig.~\ref{fig:under_results}.
The regulation error remains bounded during the sign changes and decreases with its nominal rate after the correct interval is reached. The controlled fluxes exhibit transient peaks at
the switching instants, as expected from the temporary use of an incorrect control
direction, and then settle once the scheduler recovers the stabilizing sign. Finally,
Fig.~\ref{fig:under_results} (c) confirms that both adaptive gains $N_1(\theta_1)$ and
$N_2(\theta_2)$ track the unknown switching mechanism $S(t)$, validating the use
of the vector scheduler in this reservoir-control setting.

\begin{figure}[t!]
\centering
\includegraphics[width=0.80\columnwidth,keepaspectratio]{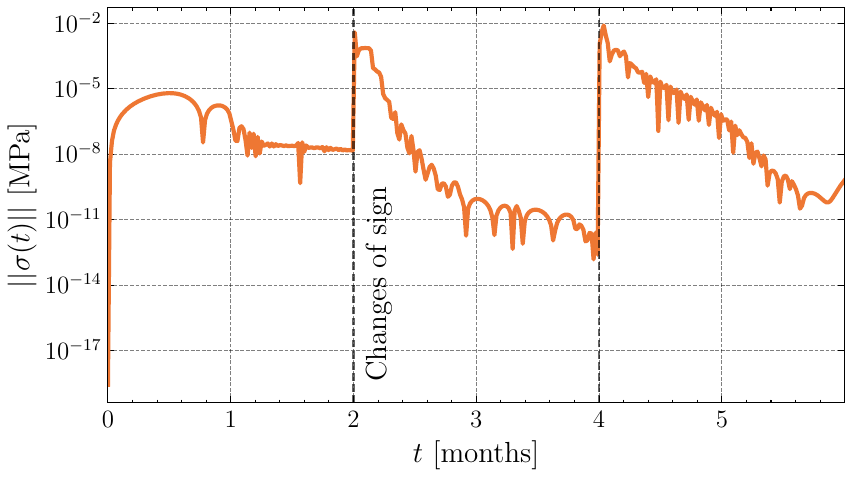}\\[-2mm]
{\small (a) Regulation error}\\

\includegraphics[width=0.80\columnwidth,keepaspectratio]{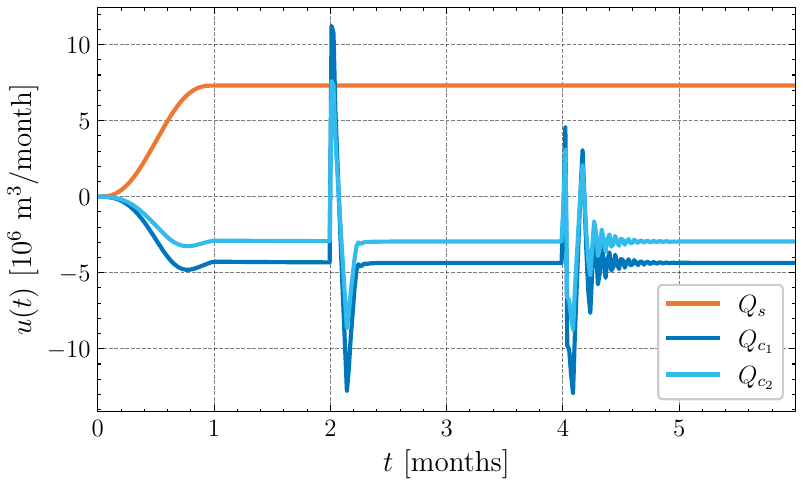}\\[-2mm]
{\small (b) Well fluxes}\\

\includegraphics[width=0.80\columnwidth,keepaspectratio]{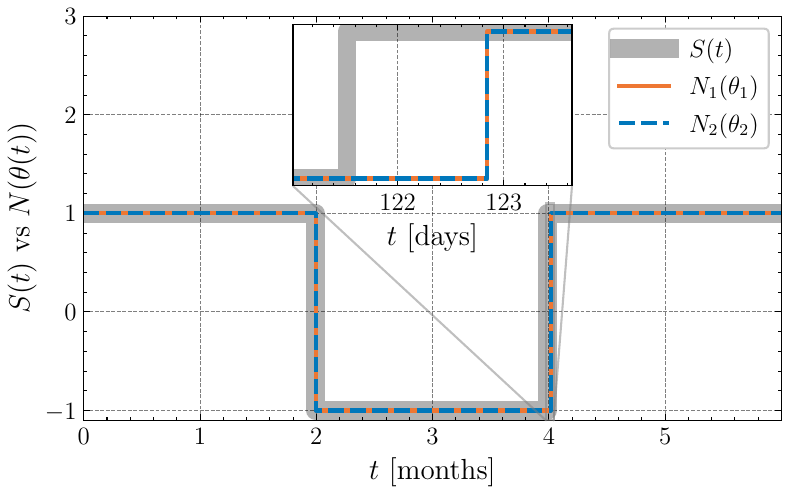}\\[-2mm]
{\small (c) Sign recovery}
\caption{Underground-reservoir example with the vector scheduler. 
(a) Evolution of the regulation error norm $\|\sigma(t)\|$. 
(b) Fluxes generated by the source term $Q_s$ and the controlled wells 
$Q_{c_1}$, $Q_{c_2}$. 
(c) Unknown switching mechanism $S(t)$ and adaptive gains 
$N_i(\theta_i)$, $i=1,2$. The inset shows the fast recovery of the correct
sign configuration after the second switching instant.}
\label{fig:under_results}
\end{figure}

\section{Conclusions}
\label{sec:Conclusions}

This paper introduced a modular switching sign-compensation approach for MIMO systems with uncertain control direction. The design preserves the nominal controller and inserts only a diagonal sign-compensation matrix between the nominal control signal and the plant. Therefore, the proposed mechanism does not modify the nominal control magnitude and does not rely on unbounded Nussbaum-type gain modulation.

The theoretical analysis was developed from a nominal Lyapunov dissipation inequality. A common degraded decay estimate was first obtained under a small-mismatch condition. Then, two scheduler mechanisms were studied. The vector scheduler provides a low-complexity channel-wise construction with an online residual-length update that makes any visited correct interval trapping. The scalar scheduler explores complete sign matrices and gives the design-time recovery guarantee through the exponential growth of the scheduler intervals, at the cost of cycling through \(2^m\) sign patterns. In both cases, once the correct sign compensation is trapped, the closed-loop dynamics coincide with the nominal ones and the original nominal stability property is recovered.

The numerical examples illustrate the two main features of the method. In the robotic visual-servoing benchmark, the scalar scheduler identifies the stabilizing sign matrix with the same number of switches as the monitoring-function approach, while reducing the integrated tracking error and the RMS control effort, at the price of a larger transient peak. In the flight-vehicle and reservoir examples, the scheduler compensates changes in the effective control direction and recovers the desired regulation after each sign change.

Future work will address coordinated vector schedulers that improve pattern reachability as in the scalar case. Furthermore, reset mechanisms for the auxiliary variables will be also investigated, which could prevent numerical overflow caused by the exponential growth of the scheduler thresholds while preserving the trapping properties after a correct sign interval is reached. 

\begin{ack}                               
Funded by the European Union. Views and opinions expressed are, however, those of the author(s) only and do not necessarily reflect those of the European Union or the European Research Council Executive Agency.  Neither the European Union nor the granting authority can be held responsible for them. This work is supported by the ERC grant INJECT, no. 101087771, doi: 10.3030/101087771.
\end{ack}

\bibliographystyle{abbrvunsrt}        
\bibliography{Bibliografias}           

\end{document}